\def\version{05 September 2022}	        	%
\documentclass[reqno,11pt]{amsart} 
\usepackage[utf8]{inputenc}
\usepackage{amsmath} 
\usepackage[mathscr]{eucal}
\usepackage{amssymb}
\usepackage{srcltx} 
\usepackage{dsfont}
\usepackage{hyperref}
\usepackage{color}
\usepackage{enumerate}
\usepackage{tikz,pgfplots}
\usetikzlibrary{arrows, positioning}
\usepackage{comment}
\usepackage{lscape}
\usepackage{graphicx}
\usepackage{tabularx}
\usepackage{caption}
\usepackage{diagbox}
\usepackage{subcaption}
\pgfplotsset{compat=1.11}
\usepgfplotslibrary{fillbetween}
\usetikzlibrary{intersections}
\pgfdeclarelayer{bg}
\pgfdeclarelayer{ft}
\pgfsetlayers{bg,main,ft}

\numberwithin{equation}{section}

\newfam\Bbbfam 
\font\tenBbb=msbm10 
\font\sevenBbb=msbm7 
\font\fiveBbb=msbm5 
\textfont\Bbbfam=\tenBbb 
\scriptfont\Bbbfam=\sevenBbb 
\scriptscriptfont\Bbbfam=\fiveBbb

\def\2{\mathbf 2}

\newcommand{\R}     {\mathbb{R}} 
 
\newcommand{\N}     {\mathbb{N}}

\def\1{{\mathchoice {1\mskip-4mu\mathrm l}      
		{1\mskip-4mu\mathrm l} 
		{1\mskip-4.5mu\mathrm l} {1\mskip-5mu\mathrm l}}} 
 
\def\comment#1{} 
\newtheoremstyle{thm}{2ex}{2ex}{\itshape\rmfamily}{} 
{\bfseries\rmfamily}{}{1.7ex}{} 

\newtheoremstyle{rem}{1.3ex}{1.3ex}{\rmfamily}{} 
{\itshape\rmfamily}{}{1.5ex}{}

\newtheorem{theorem}{Theorem}[section]

\theoremstyle{definition}

\newtheorem{remark}[theorem]{Remark}

\newcommand{\Xcal}   {{\mathcal X }}

\definecolor{Red}{rgb}{1,0,0}

\begin{document} 
	
	\title[The impact of dormancy on evolutionary branching]
 	{The impact of dormancy on evolutionary branching}
	\author[Jochen Blath, Tobias Paul, András Tóbiás and Maite Wilke Berenguer]{}
	\maketitle
	\thispagestyle{empty}
	\vspace{-0.5cm}
	
	\centerline{\sc Jochen Blath{\footnote{Goethe-Universität Frankfurt, Robert-Mayer-Straße 10, 60325 Frankfurt am Main, {\tt blath@math.uni-frankfurt.de}}}, Tobias Paul{\footnote{HU Berlin, Rudower Chaussee 25, 12489 Berlin, {\tt t.paul@math.hu-berlin.de}}}, András Tóbiás{\footnote{TU Berlin, Straße des 17. Juni 136, 10623 Berlin, {\tt tobias@math.tu-berlin.de}}} and Maite Wilke Berenguer{\footnote{HU Berlin, Rudower Chaussee 25, 12489 Berlin, {\tt mwb@math.hu-berlin.de}}}}
	\renewcommand{\thefootnote}{}
	\vspace{0.5cm}
	
	\bigskip
	
	\centerline{\small(\version)} 
	\vspace{.5cm} 
	
	\begin{quote} 
		{\small {\bf Abstract:}} 
	 	In this paper, we investigate the consequences of dormancy in the `rare mutation' and `large population' regime of stochastic adaptive dynamics. Starting from an individual-based micro-model, we first derive the polymorphic evolution sequence of the population, based on previous work by Baar and Bovier (2018). After passing to a second `small mutations' limit, we arrive at the canonical equation of adaptive dynamics, and state a corresponding criterion for evolutionary branching, extending a previous result of Champagnat and Méléard (2011). 
		
		The criterion allows a quantitative and qualitative analysis of the effects of dormancy in the well-known model of Dieckmann and Doebeli (1999) for sympatric speciation. In fact, a quite intuitive picture merges: Dormancy enlarges the parameter range for evolutionary branching, increases the carrying capacity and niche width of the post-branching sub-populations, and, depending on the model parameters, can either increase or decrease the `speed of adaptation' of populations. Finally, dormancy increases diversity by increasing the genetic distance between subpopulations.

%
	\end{quote}
	
	

	\setcounter{tocdepth}{3}
	

	\setcounter{section}{0}
	\begin{comment}{
	This is not visible.}
	\end{comment}
	
	\section{Introduction }
	\label{sec-introductionHGT}
	
	{\bf Motivation.} Dormancy is a ubiquituous trait 
	in  (but not restricted to)    microbial communities. It allows individuals to enter a reversible state of reduced metabolic activity for a limited  period of time. While it is known that dormancy   contributes    to the maintenance of microbial diversity (see e.g. \cite{JL10,LdHWB+21}), only few concrete mathematical models   seem to    have been developed to study the impact of dormancy on the evolution of species. Here, we focus on a stochastic modelling approach using the methods of adaptive dynamics, which has been employed  by Champagnat and Méléard \cite{CM11}. However, long before the advent of stochastic adaptive dynamics, the study of deterministic Lotka--Volterra systems   already laid the foundations    for the analysis of   the   evolution and interaction of species with density dependent competition.

	One   early example is the   discrete-time Lotka--Volterra model in \cite{R72}.	There, individuals live on a resource axis -- what is known in adaptive dynamics as the \emph{trait space} -- and individuals at position $x$ on the axis have an equilibrium population size $K(x)$. An interesting question arising   in this context   is the width of the niche occupied by a species, which is the space that the species occupies on the resource axis. Roughgarden found that the niche width depends on the difference of the variance of the carrying capacity and the competition function.   Precisely this   term later appears in \cite{DD99} as a criterion for \emph{evolutionary branching} in a slightly different, continuous time model.
	
	The term 'evolutionary branching' describes the splitting of a monomorphic population into a dimorphic population around a critical trait with a subsequent divergence of the traits away from the branching point. This dimorphism may be interpreted as two branches of the population which are close to each other initially but are driven apart due to  the accumulation of specific (beneficial)  mutations.  For this to happen, usually it is assumed that a monomorphic population first converges towards a critical point which, once the population has arrived, constitutes a local fitness minimum, so that the coexistence of multiple traits in its vicinity becomes possible.  Hence one may observe evolutionary branching. For a more detailed  description and  interpretation  of this mechnism  we refer to \cite[Chapter 2]{D11} and \cite[Section 2.9]{WG05}.   The theory of evolutionary branching in a non-geographic trait space as developed  by \cite{DD99} based on classical results was a milestone of adaptive dynamics modelling for sympatric speciation  -- the emergence of new species without geographic separation between populations.  However, this theory is also   seen as difficult to prove \cite{C07}.

	 Note that  Dieckmann and Doebeli provide a continuous space variant of a model introduced by \cite{CL80} in which the speciation mechanism only depends on the concentration of the competition kernel and the carrying capacity. While this result is not entirely new as we have mentioned above, it was the first time to be stated in the context of evolutionary branching and sympatric speciation. A plethora of further models with many different features which exhibit evolutionary branching have been investigated around the time (cf. \cite{GKMM98, GK00, DHK04, SMJV13} and references therein). However, we are not aware of any models in this area incorporating dormancy.
	
	The assumptions in all these models sparked a discussion on the usefulness, limitations and need of adaptive dynamics as a modelling approach. A particular criticism of Waxmann and Gavrilets in their survey article on adaptive dynamics \cite{WG05} concerns the  (limited)  parameter ranges for sympatric speciation. We will show that the introduction of a competition-induced dormancy mechanism as presented in \cite{BT20} relaxes the criterion for evolutionary branching as long as the strength of the mechanism is of the same order as the competition kernel and the carrying capacity. In particular, we claim that dormancy can favour evolutionary branching and therefore may be seen as a factor contributing to sympatric speciation (within the limitations of the existing adaptive dynamics models, i.e.\ asexual reproduction, homogeneous environment, mostly monomorphic populations etc).

	Moreover, we observe that dormancy can enrich the diversity of microbial communities in numerous other ways besides helping evolutionary branching. Firstly, we  observe larger carrying capacities.   This happens in two ways: due to  dormancy, individuals \color {black} can tolerate higher levels of competition and hence the total population size increases. Furthermore, the introduction of dormancy can help the population to occupy larger portions of the trait space. This occurs through wider branches which we also refer to as wider niches. Secondly, our dormancy mechanism can increase or decrease the speed of adaptation. This is to say that the evolution towards a fitness optimum from a lower fitness level may both be favourably or adversely impacted by dormancy depending on the model parameters. Thirdly, we can show that the (genetic) distance between subpopulations after evolutionary branching can be increased. And lastly, we observe the emergence of alternative  pathways   to dimorphic populations. 
	
	In order to make these claims mathematically sound, we will use the stochastic model of \cite{CM11} as our foundation. The polymorphic evolution sequence (PES) is the result of an additional assumption on rare mutations, i.e.\ a mutant trait fixates before a second mutation emerges. While this assumption prevents clonal interference and may alter the effects of other evolutionary features, it is one of the simplest models for polymorphic populations in the large population limit and has been adopted into several different settings before, such as cancer modelling \cite{BCM+16} or predator-prey dynamics \cite{CHL+16}. The second building stone is the \emph{canonical equation of adaptive dynamics}, which is the small mutation limit of the trait substitution sequence (the monomorphic equivalent of the PES), the latter of which was introduced in \cite{MGM+96}. This equation describes the evolution of the trait over time as a differential equation and has been discussed in  several  different models (cf. \cite{CM11,CMM13}). With these means we obtain a precise mathematical criterion for evolutionary branching.

{\bf Organisation of the paper.} 
	The remainder of the paper is structured as follows:  In Section \ref{sec-model} we present the model and necessary mathematical notation and results: Section 2.1 contains the basic stochastic individual-based model. Section 2.2 provides the corresponding polymorhic evolution sequence in the many particles and rare mutations regime. Section 2.3 considers the small mutation limit and states the evolutionary branching criterion of Champagnat and Méléard in the context of dormancy.  In Section \ref{sec-examples} we then specialize our results to the set-up of Dieckmann and Doebeli, now extended to include dormancy. We are able to provide quantitative results for the parameter ranges leading to evolutionary branching in the presence of dormancy, and also provide simulations to illustrate them (Section 3.1). In Section 3.2, we discuss further effects of dormancy for adaptive diversification, including the speed of adaptation, diversity in trait space, and potential alternative pathways to polymorphic populations. 

\section{Model and theoretical results}
\label{sec-model}
	
\subsection{A micro-model for competition-induced dormancy}
We  extend  a stochastic individual-based model with rare mutations introduced by \cite{CM11} to include dormancy.  More specifically, we assign to each individual of our population model a trait $x\in\Xcal$, where $\Xcal\subseteq\R^\ell$ with $\ell\geq 1$ is a compact set. This trait determines the behaviour of the individual. Furthermore, each trait $x$ may exhibit an active and a dormant  state (or phenotype),  which we denote by $(x,a)$ and $(x,d)$ respectively. For each $x,y\in\Xcal$ we introduce the functions
	\begin{itemize}
		\item $\lambda(x)\geq0$ is the rate at which active individuals with trait $x$ reproduce,
		\item $\mu(x)\geq 0$ is the rate at which active individuals with trait $x$ die. We assume $\mu(x)\leq\lambda(x)$ for all $x\in\Xcal$.
		\item $K\in\N$ is the carrying capacity which governs the population size,
		\item $\tfrac{\alpha(x,y)}{K}\geq 0$ is the competition kernel which determines the competition an active individual of trait $y\in\Xcal$ exerts onto an active individual of trait $x\in\Xcal$,
		\item $u_Km(x)\in(0,1]$ is the probability of mutation at birth where $u_K, m(x)\in(0,1]$,
		\item $u_K\phi(x)\in[0,1]$ is the rate of mutation of a dormant individual of trait $x$ with $\phi(x)>0$,
		\item $M((x,t),h)dh$ with $t\in \{a,d\}$ is the mutation kernel determining the law of the mutant trait $x+h\in\Xcal$ when the mutant is born from an individual with trait $x$ or is a mutant dormant individual respectively,
		\item $p(x)\in[0,1]$ is the probability with which an individual of trait $x$ can become dormant when it experiences competition,
		\item $\kappa(x)\geq 0$ is the rate at which dormant individuals of trait $x$ die. Usually, we will assume $\kappa(x)\leq \mu(x)$.
		\item $\sigma(x)>0$ is the rate at which dormant individuals of trait $x$ resuscitate,
	\end{itemize}
	Note that there are no assumptions regarding symmetry of the competition kernel or the mutation kernels. In fact, real world data from \cite{R72}  provide  an example of an asymmetric competition kernel. 
	
	
	Our  set-up   is almost a special case of the general model with phenotypic plasticity which was considered by Baar and Bovier in \cite{BB18}. 
	There, mutations can only result from births but not spontaneously, and some rates only depend on the phenotype but not the genotype. However, a simple modification of their arguments can be used to show that their results concerning the limiting behaviour of the population also apply in our case.

	\subsection{Many particles and rare mutations: The polymorphic evolution sequence}
	Under suitable assumptions regarding the behaviour of the mutation rate $u_K$ and the carrying capacity $K$ the dynamics of the population converge as $K\to\infty$ towards a pure jump Markov chain. This is due to the so-called ``rare mutations  regime'' which allows for a mutant trait to establish itself in the population before another mutation emerges. We will briefly sketch the way to our result. Denote by $N(t)$ the total population size at time $t\geq 0$ and let $N_{x,r}(t)$ with $r\in\{a,d\}$ denote the active respectively dormant population size of trait $x$ at time $t$. As $K\to\infty$, the population sizes will tend to infinity as well, therefore we consider the rescaled population and count how often each trait is represented, 
	where we denote the trait of the $i$-th individual at time $t$ and its phenotype by $z_i(t)$, i.e. $z_i(t)=(x,r)$ for some $x\in\Xcal$ and $r\in\{a,d\}$. Then we can write \[
	\nu_t^K=\frac{1}{K}\sum_{i=1}^{N_t}\delta_{z_i(t)}
	\]
	which counts the number of active/dormant individuals of each trait. If the trait space is finite, say $\Xcal=\{x_1,\ldots,x_k\}$, then we can represent $\nu_t^K$ as the vector \[
	\nu_t^K=\frac{1}{K}\begin{pmatrix}
		N_{x_1,a}(t), N_{x_1,d}(t), \cdots, N_{x_k,a}(t), N_{x_k,d}(t)
	\end{pmatrix}.
	\]
	In the absence of mutations, it is well known from \cite[Theorem 11.2.1]{EK86} that the dynamics of $N_{x_i,r}(t)/K$ for a finite number of traits $\{x_1,\ldots,x_k\}\subseteq\Xcal$ converges on finite time intervals as $K\to\infty$ to the solution $n_{x_i,r}(t)$ of the dynamical system
	\begin{equation}\begin{aligned}\label{eq: DynSystem}
			\dot{n}_{x_i,a}(t)&=n_{x_i,a}(t)\left(\lambda(x_i)-\mu(x_i)+\sigma(x_i) n_{x_i,d}(t)-\sum_{j=1}^{k}\alpha(x_i,x_j)n_{x_j,a}(t)\right)\\
			\dot{n}_{x_i,d}(t)&=p(x_i)n_{x_i,a}(t)\sum_{j=1}^{k}\alpha(x_i,x_j)n_{x_j,a}(t)-n_{x_i,d}(t)\left(\kappa(x_i)+\sigma(x_i)\right).
	\end{aligned}\end{equation}
	We say that the traits $x_1,\ldots,x_k$ can \emph{coexist}, if this system admits a coordinatewise positive and locally stable equilibrium. Denote this equilibrium by $\bar{n}=(\bar{n}_{x_1,a},\bar{n}_{x_1,d},\ldots,\bar{n}_{x_k,a},\bar{n}_{x_k,d})$. Then we know from \cite[Chapter V]{AN72} that the probability of an invading mutant trait $y\in\Xcal$ to go extinct is given by the coordinate $q_a$ for an initially active mutant and the coordinate $q_d$ for an initially dormant mutant 
  corresponding to the solution of  	
	\begin{align*}
		\lambda(y)(q_a^2-q_a)+p(y)(q_d-q_a)\sum_{j=1}^{k}\alpha(y,x_j)\bar{n}_{x_j,a}+\left(\mu(y)+(1-p(y))\sum_{j=1}^{k}\alpha(y,x_j)\bar{n}_{x_j,a}\right)(1-q_a)&=0\\
		\sigma(y)(q_a-q_d)+\kappa(y)(1-q_d)&=0
	\end{align*}
	in $[0,1]^2\setminus\{(1,1)\}$. Denote this unique solution by $q_a(y,\mathbf{x})$ for an active and similarly $q_d(y,\mathbf{x})$ for an initially dormant mutant of trait $y$. Now, we can formulate the convergence result for $\nu_t^K$.
	
	\begin{theorem}\label{Theorem: PESD}
		Fix traits $x_1,\ldots,x_k\in\Xcal$ and assume that they can coexist. Further, assume that at time $t=0$, these are the only existing traits and their population size converges as $K\to\infty$ to their equilibrium, i.e. 
		$$
		\lim\limits_{K\to\infty}\nu_0^K=\sum_{j=1}^k\bar{n}_{x_j,a}\delta_{x_j,a}+\bar{n}_{x_j,d}\delta_{x_j,d}.
		$$ 
		Moreover, assume that the mutation parameter satisfies 
		\[
		\exp(-VK)\ll u_K\ll\frac{1}{K\log(K)}\quad\text{for all }V>0\text{ as } K\to\infty,
		\]
		and that $\lambda, \alpha, \mu,\kappa$ and $\sigma$ are bounded. Lastly, assume that the solution of the dynamical system \eqref{eq: DynSystem} for coexisting traits $x_1,\ldots,x_k$ and an invading new trait $x_{k+1}\in\Xcal$ converges towards a unique locally strongly stable equilibrium $n^*(x_1,\ldots,x_k,x_{k+1})$ from any sufficiently small neighbourhood of the initial condition $(\bar{n},0,0)$. 
		
		Then the process $(\nu_{t/(Ku_K)}^K)_{t\geq 0}$ with initial condition $\nu_0^K$ converges in the sense of finite dimensional distributions  on the space of finite measures on $\Xcal\times\{a,d\}$ equipped with the weak topology as $K\to\infty$ to the jump process  $(Z_t)_{t \ge 0}$  with transitions from \[
		\sum_{j=1}^{k}\bar{n}_{x_j,a}\delta_{x_j,a}+\bar{n}_{x_j,d}\delta_{x_j,d}\qquad\text{to}\qquad \sum_{j=1}^{k+1} {n}^*_{x_j,a}\delta_{x_j,a}+n^*_{x_j,d}\delta_{x_j,d}
		\]
		at rate \begin{align*}
			&\sum_{j=1}^k\lambda(x_j)m(x_j)\bar{n}_{x_j,a}(1-q_a(x_{k+1},\mathbf{x}))M((x_j,a),x_{k+1}-x_j)\\
			+&\sum_{j=1}^k\phi(x_j)\bar{n}_{x_j,d}(1-q_d(x_{k+1},\mathbf{x}))M((x_j,d),x_{k+1}-x_j)
		\end{align*}
		and initial condition \[
		Z_0=\lim\limits_{K\to\infty}\nu_0^K=\sum_{j=1}^k\bar{n}_{x_j,a}\delta_{x_j,a}+\bar{n}_{x_j,d}\delta_{x_j,d}.
		\]
	\end{theorem}
	\begin{proof}
		The proof consists of a  straightforward modification  
		of the proof of \cite[Theorem 3.6]{BB18}.
	\end{proof}
		
\subsection{The small mutation limit and a criterion for evolutionary branching}	
	We now turn towards  obtaining a  scaling limit as the allowed radius of mutation is scaled down to $0$. For this, we need to assume that the trait space $\Xcal\subseteq\R^\ell$ is convex. Further, we introduce the scaling parameter $\varepsilon>0$ into the mutation kernel by  defining the  mutant trait by $y=x+\varepsilon Y$ where now $Y$ is distributed according to $M((x,r),h)dh$. In addition, we assume that there are no mutations in the dormant population. Then, letting $K\to\infty$, we obtain a polymorphic evolution sequence dependent on $\varepsilon$, which we call  $Z^\varepsilon=(Z_t^\varepsilon)_{t \ge 0}$, and subsequently letting $\varepsilon\to 0$ gives us convergence of $Z^\varepsilon$ to a deterministic limit. However, we need a new scaling of time to account for the smaller  size of mutations. 
	
	 As in \cite[Section 4]{CM11} we will assume a monomorphic population. Then we can associate the polymorphic evolution sequence $(Z^\varepsilon_t)_{t\geq 0}$ with a jump Markov process $(X^\varepsilon_t)_{t\geq 0}$ via $Z^\varepsilon_t=\bar{n}_{X^\varepsilon_t,a}\delta_{X^\varepsilon_t,a}+\bar{n}_{X^\varepsilon_t,d}\delta_{X^\varepsilon_t,d}$. The process $(X^\varepsilon_t)_{t\geq 0}$ now takes values in $\R^\ell$ and describes the trait of the population at time $t$. 
	\begin{theorem}
	\label{Theorem: CEAD}
		Under the assumptions of Theorem \ref{Theorem: PESD} assume that the processes $(Z^\varepsilon_t)_{t\geq 0}$ have a monomorphic initial condition 
		$$
		Z^\varepsilon_0=\bar{n}_{x,a}\delta_{x,a}+\bar{n}_{x,d}\delta_{x,d},\quad\text{i.e.}\quad X^\varepsilon_0=x
		$$ 
		for some $x\in\Xcal$. Also, assume that  the mutation rates satisfy  $\phi(x)=0$ for all $x\in\Xcal$,  and that the function \[g(y,x)=\lambda(x)m(x)\bar{n}_{x,a}(1-q_a(y,x))\] is continuous on $\Xcal^2$ and continuously differentiable with respect to $y$. Furthermore, let the map $x\mapsto M((x,a),h)dh$ be Lipschitz with respect to the Wasserstein metric on the set of probability measures and let $M((x,a),h)$ have finite and bounded in $x$ third-order moments. Then $(X^\varepsilon_t)_{t\geq 0}$ converges weakly in the space of càdlàg paths $\mathbb{D}([0,T],\R^\ell)$ on any finite time interval $[0,T]$ with $T>0$ to a deterministic function $x(t)$ 
		with $x(0)=x$ and which solves the   following  canonical equation of adaptive dynamics
		\[
		\frac{dx(t)}{dt}=\int_{\R^\ell} h[h\cdot\nabla_1g(x(t),x(t))]_+M((x(t),a),h)dh.
		\]

	\end{theorem}
	Here, $\nabla_1g(x,x)$ with $x\in\R^\ell$ denotes the gradient of $g$ with respect to the first vector. Further, the term in square brackets is to be understood as scalar product.

	\begin{proof}
		Note that we can calculate the equilibrium \[
		\bar{n}_{x,a}=\frac{(\lambda(x)-\mu(x))(\kappa(x)+\sigma(x))}{\alpha(x,x)(\kappa(x)+(1-p(x))\sigma(x))}
		\]
		and the probability of survival of the mutant trait against a monomorphic resident population
		\[
		1-q_a(y,x)=1-\frac{\alpha(y,x)\bar{n}_{x,a}(\kappa(y)+(1-p(y))\sigma(y))}{\lambda(y)(\kappa(y)+\sigma(y))}-\frac{\mu(y)}{\lambda(y)}
		\]
		explicitly. In particular, it is possible to define a function \[f(y,x)=\lambda(y)-\mu(y)-\frac{\alpha(y,x)\bar{n}_{x,a}(\kappa(y)+(1-p(y))\sigma(y))}{\kappa(y)+\sigma(y)}\] such that \[
		g(y,x)=\lambda(x)m(x)\bar{n}_{x,a}\frac{f(y,x)}{\lambda(y)}.
		\]
		Observe that this is the form of the transition function found in \cite[Theorem 4.1]{CM11} with similar properties. Therefore,  one can now continue exactly as in their proof to obtain the result.  
	\end{proof}
	
	\begin{remark}
 	In \cite{BBC17} it has been shown that one can obtain a similar result even when one takes the limits simultaneously instead of successively as we have done. However, the corresponding proof is significantly more technical. While we believe that a corresponding result with dormancy should exist, we stick with the simpler version provided here, which is sufficient for our purposes.
 	
	\end{remark}

 Our next aim is to derive an explicit   criterion for what is called \emph{evolutionary branching} in one dimension.  Recall that the term describes the situation where   at a point $x^*$, a previously monomorphic population  living  in a neighbourhood of $x^*$ may suddenly become dimorphic with an increasing distance between the genotypes. Consider now $\Xcal\subseteq\R$ as a closed interval and let $x^*$ be a point towards which the solution of the canonical equation converges. We will focus on traits which satisfy $\partial_1g(x^*,x^*)=0$ or equivalently $\partial_1 f(x^*,x^*)=0$. We refer to such traits as \emph{evolutionary singularities}. While this point will never be reached by $x(t)$, we can ask, how the stochastic system behaves around these points. A criterion for evolutionary branching was  derived  by \cite{MGM+96}. A more rigorous version was then  proved  by Champagnat and Méléard. Under the assumptions of Theorem \ref{Theorem: CEAD}, they found the following result, which carries over  one-to-one  to our model: 
		
	\begin{theorem}
	\label{thm:branching criterion}
		Consider the process $Z^\varepsilon$ and assume that the canonical equation of adaptive dynamics with initial condition $x$ converges towards an evolutionary singularity $x^*$ in the interior of $\Xcal$. We also assume that $\lambda,\mu,\kappa, \sigma$ and $p$ are three times continuously differentiable and $\alpha(x,y)$ is four times continuously differentiable with respect to the first variable. Further assume that the mutation kernel $M((x,a),\cdot)$ has mass on the positive and negative real axis for any $x$ in the interior of $\Xcal$. Assume that the function $f$ from the proof of Theorem \ref{Theorem: CEAD} satisfies \[
		\partial_{22}f(x^*,x^*)>\partial_{11}f(x^*,x^*)\quad\text{and}\quad \partial_{22}f(x^*,x^*)+\partial_{11}f(x^*,x^*)\neq 0.
		\]
		Then, \begin{itemize}
			\item if $\partial_{11}f(x^*,x^*)>0$, then there is almost surely evolutionary branching at $x^*$,
			\item if $\partial_{11}f(x^*,x^*)<0$, then there is almost surely no evolutionary branching.
		\end{itemize}
	\end{theorem}
\begin{proof}[Sketch of proof]
 Most of the proof can be taken from the proof of \cite[Theorem 4.10]{CM11}. However, we need to  specify  the fitness function for an invading mutant trait $z$ in a dimorphic population consisting of traits $x$ and $y$. This is done by setting \[
 f(z;x,y)=\lambda(z)-\mu(z)-\frac{(\alpha(z,x)\hat{n}_{x,a}+\alpha(z,y)\hat{n}_{y,a})(\kappa(z)+(1-p(z))\sigma(z))}{\kappa(z)+\sigma(z)},
 \]
 where $\hat{n}_{x,a}$ and $\hat{n}_{y,a}$ denote the active population sizes of trait $x$ and $y$ respectively in their dimorphic equilibrium. To calculate these explicitly, we consider the dynamical system \begin{equation}\begin{aligned}
 	\dot{n}_x^a(t)&=n_x^a(t)\left(\lambda(x)-\mu(x)-\alpha(x,x)n_x^a(t)-\alpha(x,y)n_y^a(t)\right)+\sigma(x) n_x^d(t)\\
 	\dot{n}_x^d(t)&=p(x)n_x^a(t)(\alpha(x,x)n_x^a(t)+\alpha(x,y)n_y^a(t))-\sigma(x) n_x^d(t)\\
 	\dot{n}_y^a(t)&=n_y^a(t)\left(\lambda(y)-\mu(y)-\alpha(y,x)n_x^a(t)-\alpha(y,y)n_y^a(t)\right)+\sigma(y) n_y^d(t)\\
 	\dot{n}_y^d(t)&=p(y)n_y^a(t)(\alpha(y,x)n_x^a(t)+\alpha(y,y)n_y^a(t))-\sigma(y) n_y^d(t),\\
 \end{aligned}\label{eq: dynamical system}
\end{equation}
 which describes the behaviour of the two traits in absence of mutations as is known from standard theory (cf. \cite{BT20}, \cite{EK86}). It is easily verified that the equilibrium of \eqref{eq: dynamical system} satisfies \[
 \alpha(y,x)\hat{n}_{x,a}+\alpha(y,y)\hat{n}_{y,a}=\frac{(\lambda(y)-\mu(y))(\kappa(y)+\sigma(y))}{\kappa(y)+(1-p(y))\sigma(y)}\]
 and \[\quad\alpha(x,x)\hat{n}_{x,a}+\alpha(x,y)\hat{n}_{y,a}=\frac{(\lambda(x)-\mu(x))(\kappa(x)+\sigma(x))}{\kappa(x)+(1-p(x))\sigma(x)}.
 \]
 Hence, $f(z,x,y)$ is well-defined if and only if the linear system has a unique solution and hence $\alpha(y,x)\alpha(x,y)\neq\alpha(x,x)\alpha(y,y)$. With this definition and explicit  representation  one can show the properties given in \cite[Proposition 4.13]{CM11}  from which the result  follows.
\end{proof}

Aside from the technical assumptions, $\partial_{11}f(x^*,x^*)$ decides whether the singularity $x^*$ is a local fitness minimum or maximum.  On the one hand, if  it is a minimum, then invading traits from the left and the right of the singularity have a higher fitness and can coexist for suitable combinations. Since these traits are now coexisting, the singular trait $x^*$ is no longer able to invade  (but traits further away may still be able to invade and push the coexisting traits further apart).   If on the other hand $\partial_{11}f(x^*,x^*)<0$ and $x^*$ is a fitness maximum, then surrounding traits of $x^*$ may invade but not extinguish the singular trait if $\partial_{11}f(x^*,x^*)+\partial_{22}f(x^*,x^*)>0$. This leads to coexistence of multiple traits in a small neighbourhood of the singular trait as was shown by \cite[Proposition 4.11]{CM11}. A graphical description of how these situations may arise is given in Figures 2.3 and 2.4 of \cite{D11}. Here, the model is slightly different but the interpretation also applies to our situation.

While the above criterion  in its abstract formulation is  unchanged from the one stated  by  Champagnat and Méléard, the implicit presence or absence of dormancy will significantly affect the emergence of evolutionary branching, as we will investigate in an important special scenario below. In general, the impact of dormancy depends heavily on the shape of the functions $\kappa$, $\sigma$ and $p$  which makes a general discussion rather complex. 

	\section{Dormancy in a classical model for evolutionary branching} 
	\label{sec-examples}

	\subsection{A simple explicit criterion for evolutionary branching in the presence of dormancy}\label{sec-3.1}

	In order to gain an understanding of the consequences of dormancy in our abstract evolutionary branching criterion, we now discuss dormancy in the concrete set-up provided by Champagnat and Méléard, which is in turn based on Dieckmann and Doebeli \cite{DD99}, Christensen and Loescke \cite{CL80}, Roughgarden \cite{R72} and references therein. 

%
%
%
%
%
	Here, the parameters  of the general micro-model from the beginning of Section \ref{sec-model} are specified to  
	\begin{gather*}
		\Xcal=[-2,2], \quad \mu(x)=0,\quad m(x)=m,\quad M((x,a),h)dh\sim\mathcal{N}(0,\rho^2),\\
		\lambda(x)=\exp\left(-\frac{x^2}{2\sigma_b^2}\right),\quad \alpha(x,y)=\exp\left(-\frac{(x-y)^2}{2\sigma_{\alpha}^2}\right).
	\end{gather*}
	We introduce  a specific type of dormancy  by letting 
	\[
	\kappa(x)=0,\quad \sigma(x)=\sigma,\quad p(x)= 1-\exp(-|x|^r/(2\sigma_p^2))	
	\]

for $\sigma_p, r >0$. Note that the  choice  $r=2$ appears  particularly natural,   since then dormancy and competition/fitness act on similar scales.


 Further, note that our choice for the dormancy-initiation function $p$  introduces a {\em reproductive trade-off}  into the system,  in the sense that individuals who can evade death by becoming dormant more efficiently  in turn   suffer from an increased need for resources and hence lower reproduction  rate. Such a trade-off has been reported in the dormancy literature (e.g. \cite{LdHWB+21}) and is a crucial modelling assumption.  

 In order to get an explicit criterion for evolutionary branching in the presence of competition-induced dormancy as realized above, we now need to calculate the derivatives of the fitness  function $f$ and then apply Theorem \ref{thm:branching criterion}. We obtain  
	\[
	f(y,x)=\lambda(y)-\frac{\alpha(y,x)\lambda(x)}{1-p(x)}(1-p(y)).
	\]
	Calculating the first derivative  yields   
	\begin{align}
		\partial_1f(x,x)=\lambda'(x)+\frac{\lambda(x)\alpha(x,x) p'(x)}{1-p(x)}-\partial_1\alpha(x,x)\lambda(x)=\lambda'(x)+\frac{\lambda(x)\alpha(x,x) p'(x)}{1-p(x)}, 
		\label{eq: first-deriv}
	\end{align}	
	and the second derivative with respect to the first component gives \begin{align}
		\partial_{11}f(x,x)&=\lambda''(x)-\partial_{11}\alpha(x,x)\lambda(x)+2\frac{\lambda(x)}{1-p(x)}\partial_1\alpha(x,x)p'(x)+\frac{\lambda(x)\alpha(x,x)p''(x)}{1-p(x)}\nonumber\\
		&=\lambda''(x)-\partial_{11}\alpha(x,x)\lambda(x)+\frac{\lambda(x)\alpha(x,x)p''(x)}{1-p(x)},\label{eq: second-deriv}
	\end{align}
	 Note that more generally,  these equations hold as long as we assume the competition  function to be at a maximum for identical traits.

		For our choice of $p$, the  first  derivative takes the form
\[
	\partial_1f(x,x)=-\frac{xe^{-x^2/(2\sigma_b^2)}}{\sigma_b^2}+\frac{rx|x|^{r-2}e^{-x^2/(2\sigma_b^2)}}{2\sigma_p^2}=0\quad\Longleftrightarrow\quad \frac{rx|x|^{r-2}}{2\sigma_p^2}=\frac{x}{\sigma_b^2}.
	\]
	For $r>1$, this admits the solution $x=0$ but if $r\neq 2$ there are up to two further solutions which can be computed in the closed form: 
	\[
	x_{1,2}=\pm\left(\frac{r\sigma_b^2}{2\sigma_p^2}\right)^{-1/(r-2)}.
	\]
	These two points are also the only equilibria in the case $r\in(0,1]$. The stability of these zeros of the first derivative depends only on $r$ (for $r\neq 2$). For $r>2$, only $0$ is stable, whereas for $r<2$, the equilibrium $0$ becomes unstable and $x_{1,2}$ both become stable.
	 
	 In the most interesting case  when $r=2$, the equilibrium $0$ is stable if and only if  the {\em stability condition}  
	$$
	\frac{1}{\sigma_p^2}<\frac{1}{\sigma_b^2}
	$$ 
  is satisfied.  
	
	The second derivative in  our  scenario reads 
	\begin{align*}
		\partial_{11}f(x,x)=\frac{(x^2-\sigma_b^2)e^{-x^2/(2\sigma_b^2)}}{\sigma_b^4}+\frac{e^{-x^2/(2\sigma_b^2)}}{\sigma_\alpha^2}-\frac{r|x|^r(r|x|^r-2\sigma_p^2r+2\sigma_p^2)e^{-x^2/(2\sigma_b^2)}}{4x^2\sigma_p^4}.
	\end{align*}
	If now $x^*=0$ is the stable evolutionary singularity, then this simplifies to 
	\[
	\partial_{11}f(0,0)=\begin{cases}
		\frac{1}{\sigma_{\alpha}^2}-\frac{1}{\sigma_b^2},&\quad\text{if }r>2\\
		\frac{1}{\sigma_{\alpha}^2}-\frac{1}{\sigma_b^2}+\frac{1}{\sigma_p^2},&\quad\text{if }r=2.
	\end{cases}
	\]
	We also compute
	\[
	\partial_{22}f(0,0)=\begin{cases}
		\frac{1}{\sigma_{\alpha}^2}+\frac{1}{\sigma_b^2},&\quad\text{if }r>2\\
		\frac{1}{\sigma_{\alpha}^2}+\frac{1}{\sigma_b^2}-\frac{1}{\sigma_p^2},&\quad\text{if }r=2.
	\end{cases}
	\]
	The criterion for evolutionary branching in the case of $r<2$ is also dependent on the equilibrium itself and as such the relationship between the parameters becomes more involved. For now, we only consider the case $r\geq 2$. We find that $\partial_{22}f(0,0)>\partial_{11}f(0,0)$ is always satisfied. Comparing our criterion for evolutionary branching around $0$ with the criterion obtained by \cite{CM11}, we see no change in the case $r>2$. This is due to the effect of dormancy being too weak in the neighbourhood of the singularity and in fact being of different order compared to the birth and competition rates. However, when they are of the same order in the case $r=2$, then we obtain an additional positive constant depending on the dormancy parameter	$\sigma_p$. This shows that evolutionary branching is supported by dormancy, as long as the mechanism is sufficiently strong. At the same time, dormancy must not be too strong, i.e. $\tfrac{1}{\sigma_p^2}>\tfrac{1}{\sigma_b^2}$, because the population would not converge to the evolutionary singularity $0$ but instead be driven towards the boundary of the trait space. 
 
Hence we arrive at the main result of this section.
\begin{theorem}[Evolutionary branching in the presence of dormancy]
\label{thm:criterion}
Under the above notation, with $r=2$, assume that the stability condition 
$$
	\frac{1}{\sigma_p^2}<\frac{1}{\sigma_b^2}
	$$ 
is satisfied. Then, we observe local convergence into the stable singularity at 0 with subsequent evolutionary branching, if 
$$
\frac{1}{\sigma_{\alpha}^2}-\frac{1}{\sigma_b^2}+\frac{1}{\sigma_p^2}>0.
$$
In particular, the presence of dormancy always increases the parameter range for evolutionary branching.
\end{theorem}

	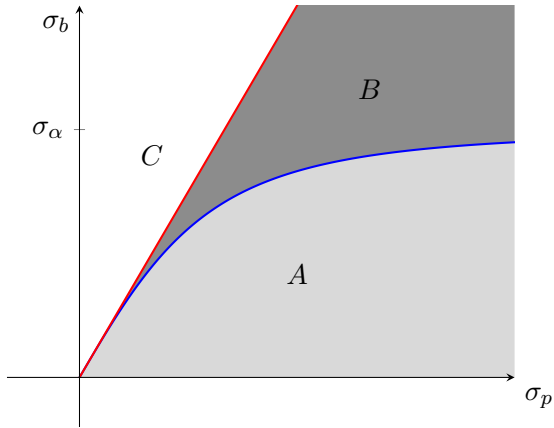
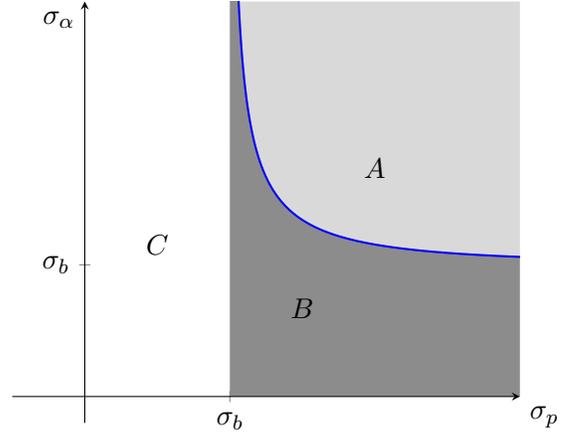
\begin{figure}
		\begin{subfigure}[t]{0.45\textwidth}
			\resizebox{\textwidth}{!}{
				\begin{tikzpicture}[declare function={a=0.5;b=3;f(\x)=1/((1+1/\x^2)^.5); g(\x)=\x;}]
					\begin{axis}[axis lines=middle,axis on top,xlabel=$\sigma_p$,ylabel=$\sigma_b$,
						xmin=-0.5,xmax=3,ymin=-0.2,ymax=1.5,ytick={1},yticklabels={$\sigma_\alpha$},
						xtick=\empty, xticklabels={},
						every axis x label/.style={at={(current axis.right of origin)},anchor=north west},
						every axis y label/.style={at={(current axis.above origin)},anchor=north east}]
						\addplot[name path=A,blue,thick,domain=0.0001:5,samples=1000] {f(x)};
						\addplot[name path=B,red,thick,domain=0.0001:5,samples=200] {g(x)};
						\path[name path=C] (\pgfkeysvalueof{/pgfplots/xmin},0) -- (\pgfkeysvalueof{/pgfplots/xmax},0);
						\addplot [gray!90] fill between [of=B and A,soft clip={domain=0.001:3},
						];
						\addplot [gray!30] fill between [of=C and A,soft clip={domain=0:b},
						];
						\path ({(b)/2},{f((b)/2)/2}) node{$A$}
						({(1+b)/2},{1.3*f((1+b)/2)}) node{$B$}
						({(1)/2},{1*f((1+b)/2)}) node{$C$};
					\end{axis}
				\end{tikzpicture}
			}
			\subcaption{The impact of $\sigma_p$ and $\sigma_b$ for constant $\sigma_\alpha$ on evolutionary branching in the case $r=2$.}
		\end{subfigure}\hfill
		\begin{subfigure}[t]{0.45\textwidth}
			\resizebox{\textwidth}{!}{
				\begin{tikzpicture}[declare function={a=0.5;b=3;f(\x)=1/((1-1/\x^2)^.5);}]
					\begin{axis}[axis lines=middle,axis on top,xlabel=$\sigma_p$,ylabel=$\sigma_\alpha$,
						xmin=-0.5,xmax=3,ymin=-0.2,ymax=3,xtick={1},xticklabels={$\sigma_b$},
						ytick={1},yticklabels={$\sigma_b$},
						every axis x label/.style={at={(current axis.right of origin)},anchor=north west},
						every axis y label/.style={at={(current axis.above origin)},anchor=north east}]
						\addplot[name path=A,blue,thick,domain=0.00001:5,samples=1000] {f(x)};
						\path[name path=B] (\pgfkeysvalueof{/pgfplots/xmin},\pgfkeysvalueof{/pgfplots/ymax}) -- (\pgfkeysvalueof{/pgfplots/xmax},\pgfkeysvalueof{/pgfplots/ymax});
						\path[name path=C] (\pgfkeysvalueof{/pgfplots/xmin},0) -- (\pgfkeysvalueof{/pgfplots/xmax},0);
						\addplot [gray!30] fill between [
						of=A and B,soft clip={domain=1:3},
						];
						\addplot [gray!90] fill between [
						of=A and C,soft clip={domain=1:b},
						];
						\path ({(b)/2},{f((b)/2)/2}) node{$B$}
						({(1+b)/2},{1.5*f((1+b)/2)}) node{$A$}
						({(1)/2},{1*f((1+b)/2)}) node{$C$};
					\end{axis}
				\end{tikzpicture}
			}
			\subcaption{The impact of $\sigma_p$ and $\sigma_\alpha$ for constant $\sigma_b$ on evolutionary branching in the case $r=2$.}
		\end{subfigure}
		\caption{The blue curve (showing $\partial_{11}f(0,0)=0$) decides whether there is evolutionary branching, the red curve (showing $\tfrac{1}{\sigma_p^2}=\tfrac{1}{\sigma_b^2}$) decides whether $0$ is a stable evolutionary singularity. In area $A$ (light grey), $0$ is a stable singularity, but there is no evolutionary branching. Area $B$ (dark grey) shows the admissible combinations such that $0$ is a stable singularity and the evolutionary branching criterion $\partial_{11}f(0,0)>0$ is satisfied. In area $C$, $0$ is not a stable singularity. We omit the comparision of $\sigma_{\alpha}$ against $\sigma_b$ because it is qualitatively similar to Figure (a). }
		\label{Fig: BranchingDynamics-1}
	\end{figure}
	
We illustrate this behaviour in Figure \ref{Fig: BranchingDynamics-1}	
	and with  concrete simulations of the stochastic population model. For this, we have chosen $r=2$, $\sigma_p^2=2$, $\sigma_\alpha^2=2$ and different $\sigma_b$. The results are shown in Figure \ref{Fig: Expdormsims-r=2}. Note that in the case $\sigma_b^2=1.2$, there would be no evolutionary branching without the aid of dormancy since $\tfrac{1}{\sigma_{\alpha}^2}-\tfrac{1}{\sigma_b^2}<0$.
	\begin{figure}
		\begin{subfigure}[t]{0.45\textwidth}
			\hspace{-0.75cm}\begin{tikzpicture}
				\node (img1)  {\includegraphics[width=1.2\textwidth]{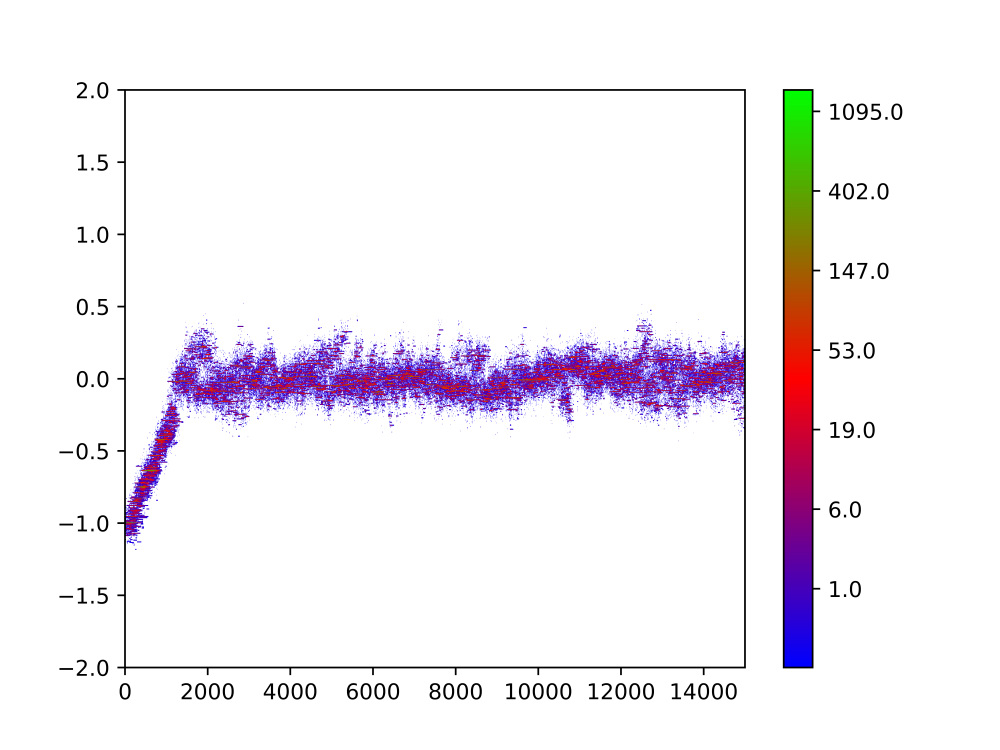}};
				\node[below=of img1, node distance=0cm, yshift=1.5cm, xshift=-0.5cm] {\footnotesize Time};
				\node[left=of img1, node distance=0cm, rotate=90, anchor=center,yshift=-1.5cm] {\footnotesize Trait};
			\end{tikzpicture}
			\subcaption{A simulation of the model with $r=2$, $\sigma_\alpha^2=2$, $\sigma_b^2=0.8$, $\sigma_p^2=2$, $\sigma=0.1$, $m=0.01$, $\rho^2=0.05$, $K=1000$.}
		\end{subfigure}\hspace{0.7cm}
		\begin{subfigure}[t]{0.45\textwidth}
			\hspace{-0.5cm}\begin{tikzpicture}
				\node (img1)  {\includegraphics[width=1.2\textwidth]{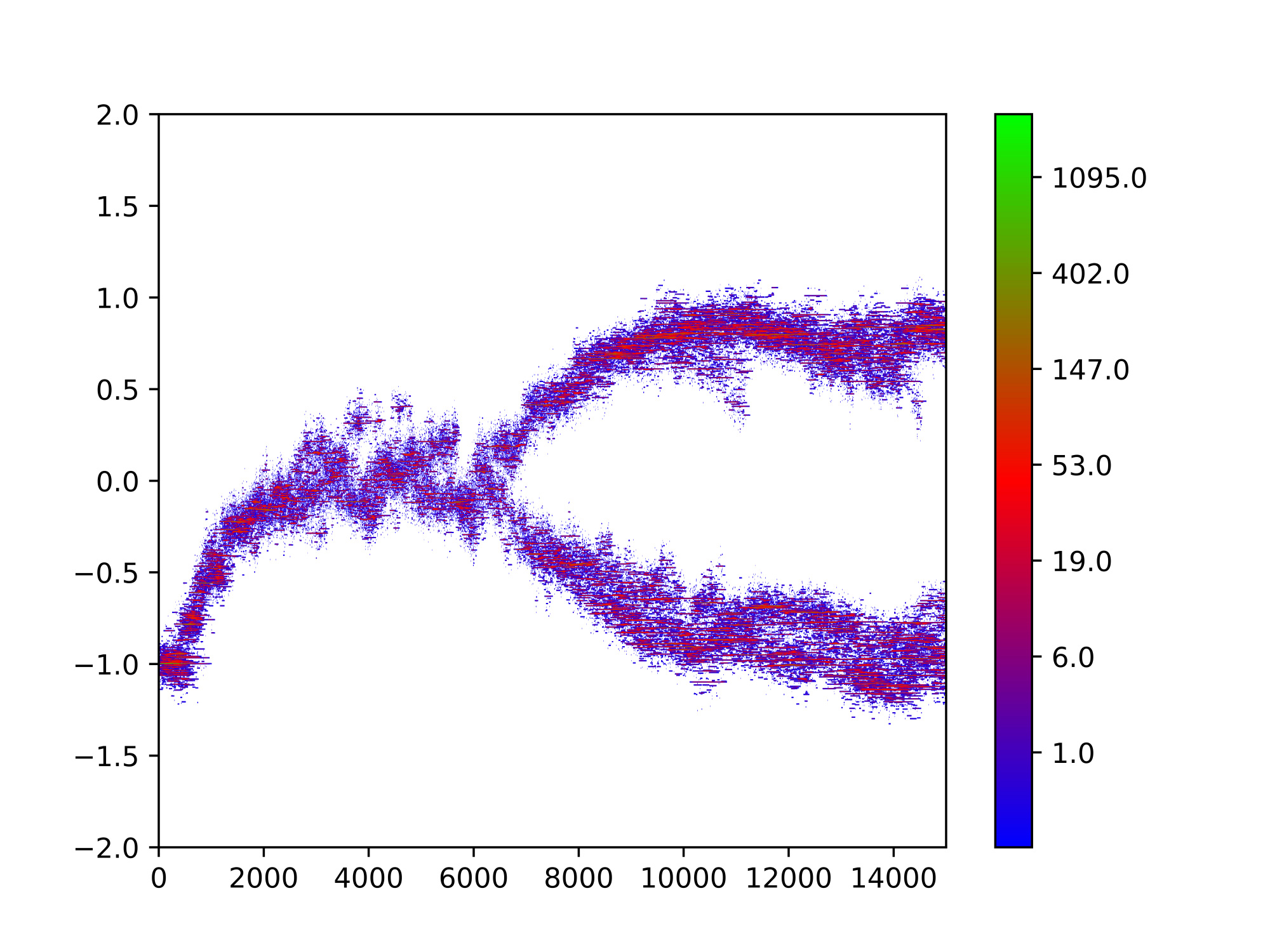}};
				\node[below=of img1, node distance=0cm, yshift=1.5cm, xshift=-0.5cm] {\footnotesize Time};
				\node[left=of img1, node distance=0cm, rotate=90, anchor=center,yshift=-1.5cm] {\footnotesize Trait};
			\end{tikzpicture}
			\subcaption{A simulation of the model with $r=2$, $\sigma_\alpha^2=2$, $\sigma_b^2=1.2$, $\sigma_p^2=2$, $\sigma=0.1$, $m=0.01$, $\rho^2=0.05$, $K=1000$.}
		\end{subfigure}
		\caption{Simulations showing the population size of each trait over time. The trait space is on the vertical line with the size of each trait at a given time being indicated by a colour corresponding to the scale next to the image. Initially, the population is composed of $250$ active individuals with trait $-1$. }
		\label{Fig: Expdormsims-r=2}
	\end{figure}
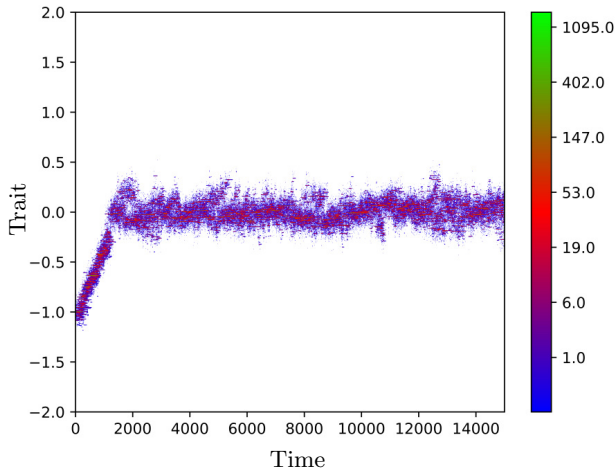
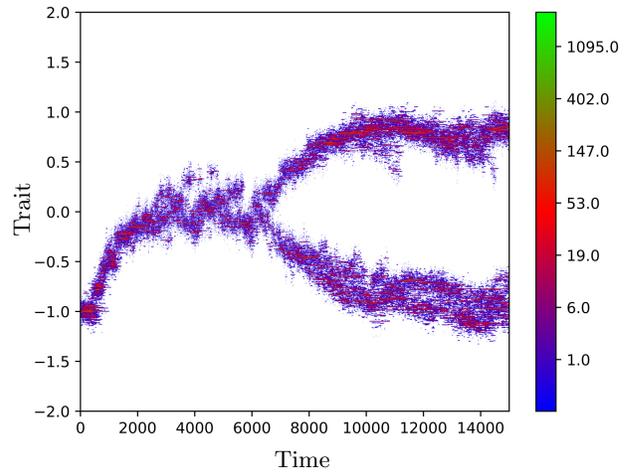

	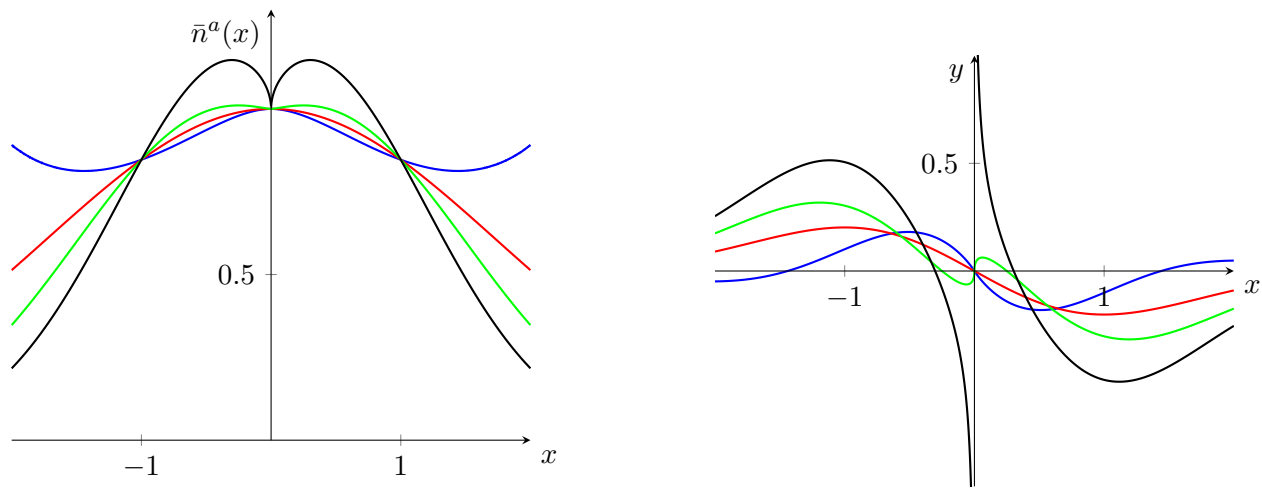
\begin{figure}
		\begin{subfigure}[t]{0.45\textwidth}
			\resizebox{\textwidth}{!}{
				\begin{tikzpicture}[declare function={a=0.5;b=3; c=1;f(\x)=exp(-\x^2/(2))/exp(-abs(\x)^2.5/(3));g(\x)=exp(-\x^2/(2))/exp(-abs(\x)^2/(3));h(\x)=exp(-\x^2/(2))/exp(-abs(\x)^1.5/(3));m(\x)=exp(-\x^2/(2))/exp(-abs(\x)^.5/(3));}]
					\begin{axis}[axis lines=middle,axis on top,xlabel=$x$,ylabel=$\bar{n}^a(x)$,
						xmin=-2,xmax=2,ymin=-0,ymax=1.3,ytick={0.5},yticklabels={$0.5$},
						xtick={-1,0,1}, xticklabels={$-1$,$0$,$1$},
						every axis x label/.style={at={(current axis.right of origin)},anchor=north west},
						every axis y label/.style={at={(current axis.above origin)},anchor=north east}]
						\addplot[name path=A,blue,thick,domain=-2:2,samples=2000] {f(x)};
						\addplot[name path=A,red,thick,domain=-2:2,samples=2000] {g(x)};
						\addplot[name path=A,green,thick,domain=-2:2,samples=2000] {h(x)};
						\addplot[name path=A,black,thick,domain=-2:2,samples=2000] {m(x)};
					\end{axis}
				\end{tikzpicture}
			}
			\subcaption{Examples for the equilibrium population sizes of a single trait in absence of competition for $\sigma_{b}^2=1$ and $\sigma_p^2=1.5$. The values of $r$ are $2.5$ (blue), $2$ (red), $1.5$ (green) and $0.5$ (black).}
			\label{Fig: Equilibrium}
		\end{subfigure}\hfill
		\begin{subfigure}[t]{0.45\textwidth}
			\resizebox{\textwidth}{!}{
				\begin{tikzpicture}[declare function={a=0.5;b=3; c=1;p(\x)=exp(-\x^2/(2))*((2.5*\x*abs(\x)^(2.5-2))/3-\x/1);g(\x)=exp(-\x^2/(2))*((2*\x*abs(\x)^(2-2))/3-\x/1);h(\x)=exp(-\x^2/(2))*((1.5*\x*abs(\x)^(1.5-2))/3-\x/1);m(\x)=exp(-\x^2/(2))*((0.5*\x*abs(\x)^(0.5-2))/3-\x/1);}]
					\begin{axis}[axis lines=middle,axis on top,xlabel=$x$,ylabel=$y$,
						xmin=-2,xmax=2,ymin=-1,ymax=1,ytick={0.5},yticklabels={$0.5$},
						xtick={-1,0,1}, xticklabels={$-1$,$0$,$1$},
						every axis x label/.style={at={(current axis.right of origin)},anchor=north west},
						every axis y label/.style={at={(current axis.above origin)},anchor=north east}]
						\addplot[name path=A,blue,thick,domain=-2:2,samples=2000] {p(x)};
						\addplot[name path=A,red,thick,domain=-2:2,samples=2000] {g(x)};
						\addplot[name path=A,green,thick,domain=-2:2,samples=2000] {h(x)};
						\addplot[name path=A,black,thick,domain=-2:-0.0001,samples=1000] {m(x)};
						\addplot[name path=A,black,thick,domain=0.0001:2,samples=1000] {m(x)};
					\end{axis}
				\end{tikzpicture}
			}
			\caption{Examples for the derivative of the fitness function $y=\partial_1f(x,x)$ for $\sigma_{b}^2=1$ and $\sigma_p^2=1.5$. The values of $r$ are $2.5$ (blue), $2$ (red), $1.5$ (green) and $0.5$ (black).}
			\label{Fig: Derivative}
		\end{subfigure}
		\caption{Equilibrium population sizes and gradients of the fitness function for various choices of $r$.}
	\end{figure}
	

\begin{remark}[The impact of the parameter $r$]

We display  in Figure \ref{Fig: Equilibrium} a selection of fitness landscapes for different parameters of $r$. 
We see that in the case $r\geq 2$, the stable singularity $0$ is a local maximum of the equilibrium population. However, as soon as $r<2$, this former local maximum turns into a local minimum with two new local maxima emerging symmetrically around $0$. These are the new stable singularities that we computed previously. Regarding the stability, we see that for $r>2$, the singularity $0$ may not be globally stable in the sense that the solution of the canonical equation of adaptive dynamics converges to $0$ from any starting trait in the trait space. This effect becomes even more pronounced if we plot the partial derivative $\partial_1f(x,x)$ with the same parameters as is depicted in Figure \ref{Fig: Derivative}. However, in the critical case $r=2$, there seems to be a dichotomy between either an unstable singularity or a globally attracting singularity at $0$. For $r<2$, each of the two singularities are stable in the sense that there is convergence to the singularity which is closest to the starting trait. Heuristically, as long as $r\in(1,2)$, the population could cross the equilibrium $0$ as long as the driving forces of evolutionary branching are sufficiently strong, since there is only little resistance to overcome (boundedness of $\partial_{1}f(x,x)$ around $0$). However, if $r\leq 1$, then the derivative becomes unbounded and the fitness valley can only be crossed by a sufficiently large mutation.
	
	Already the fact that there are two stable singularities indicates that we should be able to observe another enrichment in species diversity. In fact, in our model mutations are normally distributed on the trait space, so even if there is no evolutionary branching in the sense of our criterion $\partial_{11}f(x^*,x^*)>0$ at the equilibrium $x^*$, we still can -- and eventually will -- obtain mutations which are closer to the second equilibrium. The traits at the two equilibria can coexist, since they have the same fitness but the competition may drive the branches further apart.
\end{remark}

	\begin{remark}[More general dormancy functions $p$]
 		 
It is easy to see that a dormancy mechanism with constant $\sigma$ and no death in the dormant state always extends the classical parameter range for evolutionary branching, if $p$ is convex around the singularity, since then the last term in \eqref{eq: second-deriv} is strictly positive.  
	In the setting described in \cite[Chapter 2]{D11}, the fitness function around a singularity can also be thought of as a trade-off function. The forces taken into account are the birth rate and the competition rates. Moving away from an optimal state requires the reduction in reproduction to be compensated by lower competition or vice versa an increase in competitive pressure to be compensated by increased birth rates. If $\partial_{11}f(x^*,x^*)>0$, then this shows a convex trade-off between the function governing births and the function determining competition around the singularity. By this we mean that the loss in reproduction is outweighed by a gain in reduced competition.
	Since dormancy may be seen as another decrease in competitive pressure by allowing individuals to escape into a dormant state, the shape of the dormancy initiation probability impacts the shape of the fitness function directly. If $p$ if convex, then it increases the curvature of $f$, while if it is concave, it decreases the curvature. However, this may still be difficult to investigate, since we can see from equation \eqref{eq: first-deriv} that $p$ also affects the position of the singularity. Suppose that $x^*$ is a singularity in the model without dormancy. The simplest comparison in this situation can be made, when we assume $p(x^*)$ to be a local minimum or maximum. Then the singularity is present in both settings and having a local minimum of $p$ at $x^*$ benefits evolutionary branching around this point, while a local maximum reduces the parameter range for branching.
	\end{remark}
	
\subsection{Further consequences of competition-induced dormancy.}	
Now we investigate some further aspects of dormancy (apart from the criterion for evolutionary branching that we determined in Section~\ref{sec-3.1}), which we will justify partially mathematically and partially only heuristically or via simulations.
	
{\bf Dormancy can increase the number of subsequent branchings.}	
  Note that our criterion for evolutionary branching only applies for the first branching since after that point, we no longer have the approximation of the process via the canonical equation. In order to investigate the impact of dormancy on subsequent branching events, we use the example parameters shown in Figure \ref{Fig: Expdormsims-r=2-dist}. We highlight that in both cases -- with and without dormancy -- the criterion for evolutionary branching is satisfied for the evolutionary singularity $0$. However, in the setting with dormancy, an additional branch emerges compared to the simulation with identical parameters but without dormancy. This shows that dormancy may also favour evolutionary branching in a setting which we cannot describe accurately with our approach.

\begin{figure}
	\begin{subfigure}[t]{0.45\textwidth}
		\hspace{-0.75cm}\begin{tikzpicture}
			\node (img1)  {\includegraphics[width=1.2\textwidth]{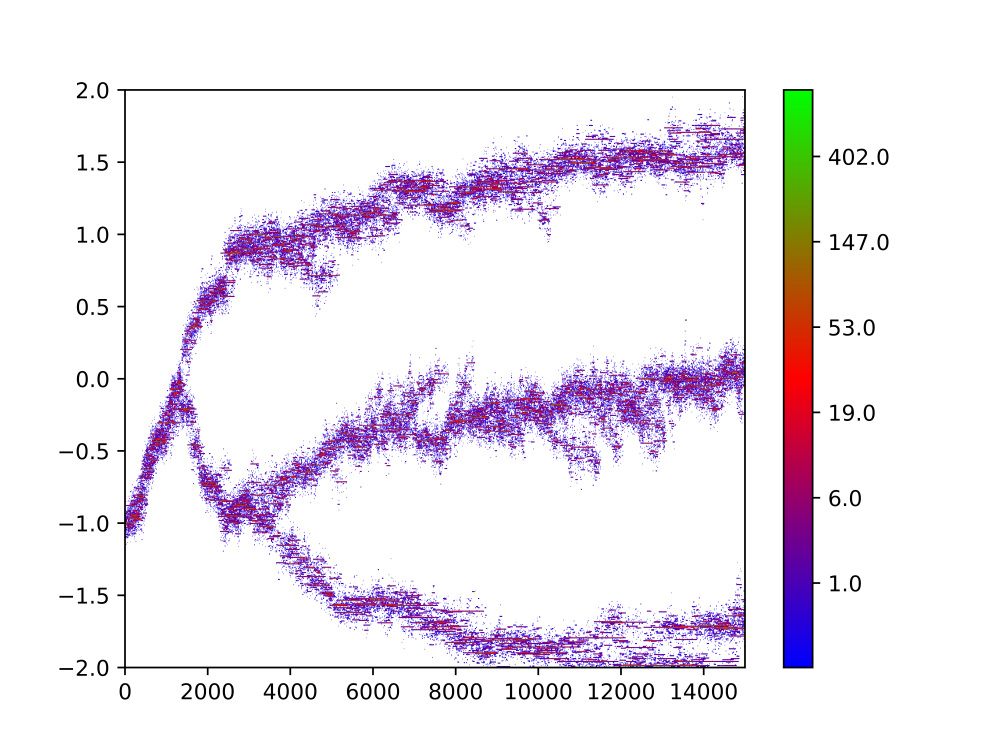}};
			\node[below=of img1, node distance=0cm, yshift=1.5cm, xshift=-0.5cm] {\footnotesize Time};
			\node[left=of img1, node distance=0cm, rotate=90, anchor=center,yshift=-1.5cm] {\footnotesize Trait};
		\end{tikzpicture}
		\subcaption{A simulation of the model with $r=2$, $\sigma_\alpha^2=1$, $\sigma_b^2=1.5$, $\sigma_p^2=3$, $\sigma=1$, $m=0.01$, $\rho^2=0.05$, $K=1000$.}
	\end{subfigure}\hspace{0.7cm}
	\begin{subfigure}[t]{0.45\textwidth}
		\hspace{-0.5cm}\begin{tikzpicture}
			\node (img1)  {\includegraphics[width=1.2\textwidth]{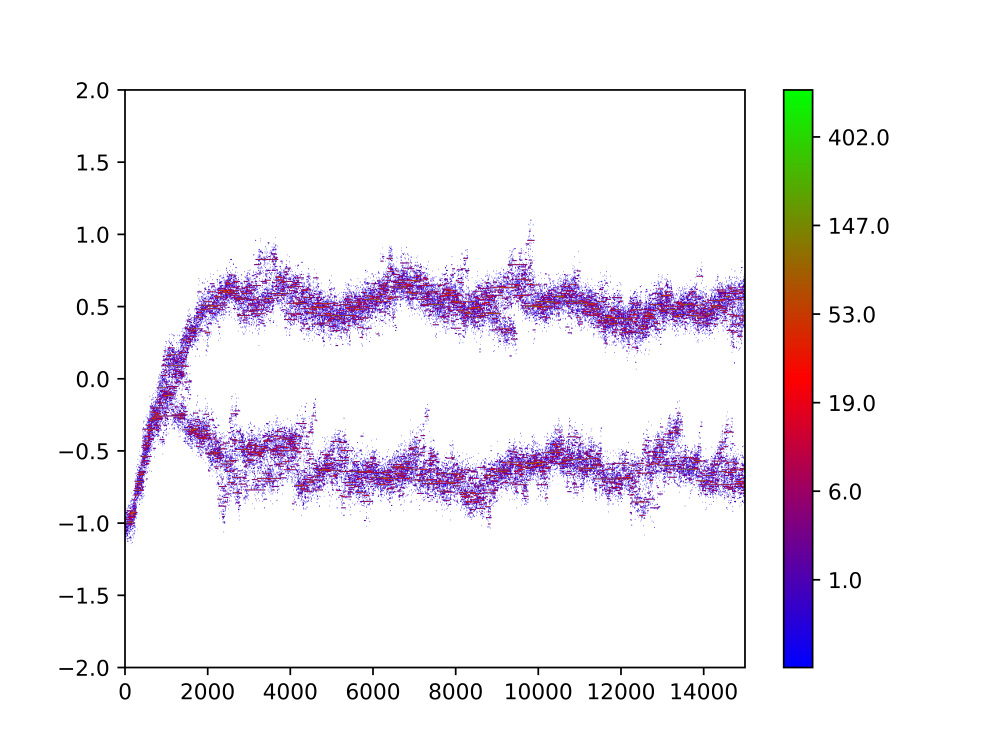}};
			\node[below=of img1, node distance=0cm, yshift=1.5cm, xshift=-0.5cm] {\footnotesize Time};
			\node[left=of img1, node distance=0cm, rotate=90, anchor=center,yshift=-1.5cm] {\footnotesize Trait};
		\end{tikzpicture}
		\subcaption{A simulation of the model without the dormancy mechanism and $\sigma_\alpha^2=2$, $\sigma_b^2=1.5$, $m=0.01$, $\rho^2=0.05$, $K=1000$.}
	\end{subfigure}
	\caption{Simulations showing the population size of each trait over time. Initially, the population is composed of $250$ active individuals with trait $-1$. The distance between the branches before secondary branchings is significantly increased on the left.}
	\label{Fig: Expdormsims-r=2-dist}
\end{figure}
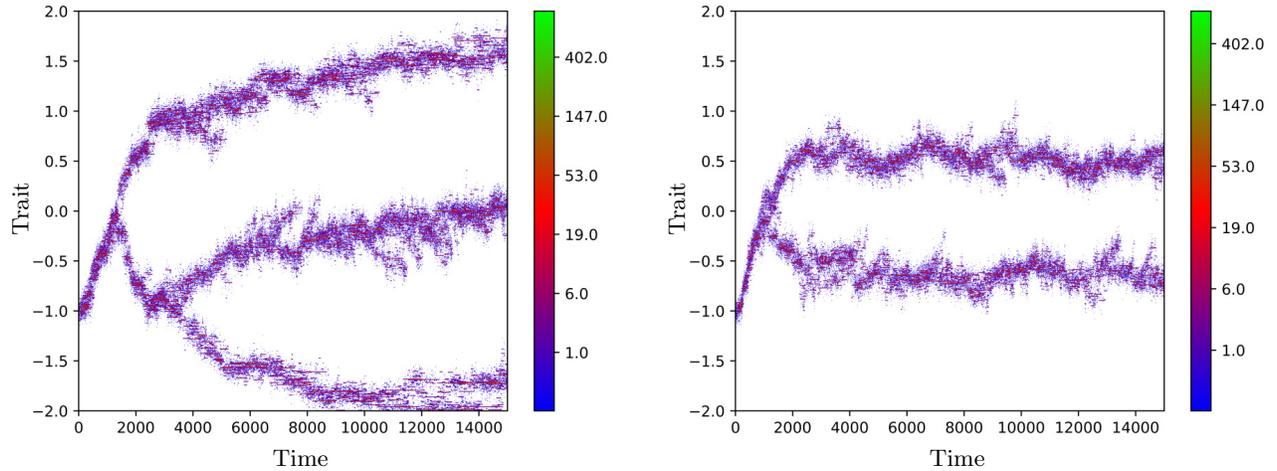
	
{\bf Dormancy can both decrease and increase the speed of adaptation.}	
	To investigate the speed of adaptation (that is, the speed with which sub-populations reach equilibria in trait space) we focus on the case $r=2$. To this aim, we compare the transition rates of the PES with and without dormancy. Suppose that in both models all rates are equal except for the additional dormancy mechanism and that the population is monomorphic with resident trait $x\in\Xcal$. Then the rate at which a mutant trait $y$ invades is given by \[
	\lambda(x)m(x)\cdot\frac{\lambda(x)-\mu(x)}{\alpha(x,x)}\cdot (1-q(y,x)) M(x,y-x)
	\]
	in the model without dormancy, where $q(y,x)$ is the probability of extinction of a single individual of trait $y$ against trait $x$ and $M(x,y-x)$ is the mutation kernel governing the probability of obtaining trait $y$ from a mutant offspring of trait $x$. In the setting with dormancy where mutations only result from births, we obtain \[
	\lambda(x)m(x)\cdot\frac{\lambda(x)-\mu(x)}{\alpha(x,x)(1-p(x))}\cdot (1-q_a(y,x)) M((x,a),y-x).
	\]
	as the rate of invasion of a mutant trait $y$.
	In our concrete example, recall $\mu(x)=0$ and $\alpha(x,x)=1$. Since the rate of birth, mutation probability and mutation kernel are assumed to be equal, we need to investigate the relationship between \[
	s(y,x,a)=\frac{\lambda(x)-\mu(x)}{\alpha(x,x)}\cdot (1-q(y,x))=\lambda(x)-\frac{\alpha(y,x)\lambda^2(x)}{\lambda(y)}\]
	and 
	\[ s(y,x,d)=\frac{\lambda(x)-\mu(x)}{\alpha(x,x)(1-p(x))}\cdot (1-q_a(y,x))=\frac{\lambda(x)}{1-p(x)}-\frac{\alpha(y,x)\lambda^2(x)(1-p(y))}{\lambda(y)(1-p(x))^2}.
	\]
	Plugging in our rates and letting $y=x+\varepsilon$ yields \[
	\frac{s(x+\varepsilon,x,d)}{s(x+\varepsilon,x,a)}=e^{x^2/(2\sigma_p^2)}\cdot\frac{1-\exp(-\frac{\varepsilon^2}{2\sigma_\alpha^2}-\frac{x^2}{2\sigma_b^2}+\frac{(x+\varepsilon)^2}{2\sigma_b^2}+\frac{x^2}{2\sigma_p^2}-\frac{(x+\varepsilon)^2}{2\sigma_p^2})}{1-\exp(-\frac{\varepsilon^2}{2\sigma_\alpha^2}-\frac{x^2}{2\sigma_b^2}+\frac{(x+\varepsilon)^2}{2\sigma_b^2})}\xrightarrow{\varepsilon\to 0} e^{x^2/(2\sigma_p^2)}\left(1-\frac{\sigma_b^2}{\sigma_p^2}\right).
	\]
	In particular, in the small mutation limit the rate of evolution is faster in the model with dormancy if and only if \[
	1-\frac{\sigma_{b}^2}{\sigma_p^2}>e^{-x^2/(2\sigma_p^2)}\quad\Longleftrightarrow\quad |x|>\sqrt{2\sigma_p^2\log\left(\frac{\sigma_p^2}{\sigma_p^2-\sigma_b^2}\right)}.
	\]
	Hence, close to the singularity $0$, the introduction of our dormancy mechanism slows down the evolutionary process. This can also be seen in Figure \ref{Fig: Expdormsims-r=2-dist}. The forces acting against each other are an increased rate of mutation from a larger equilibrium population size and the fact that a larger population size decreases the probability of a successful invasion due to increased competition. Far from the optimal population size, more frequent mutations outweigh more competition since population sizes are small.
		
{\bf Dormancy can increase the diversity in trait space.}		
 If branches are pushed further apart, we consider this as an `increase in diversity' in trait space.
 While we have seen that the introduction of dormancy may facilitate additional branching, we will investigate the location of the branches after the first branching event and before any further branchings have occurred. The simplest case is given by symmetric branching around the singularity $0$, so we will assume the parameters to be as in Figure \ref{Fig: Expdormsims-r=2-dist}. Note that we can assume symmetric branching since all of the functions are symmetric around $0$. Let us assume that the population in the rare mutation limit only consists of the two coexisting traits $x_1>0>x_2$. Due to symmetry, we expect $x:=x_1=-x_2$. One could then ask about the distance between the coexisting traits after they have settled into their equilibrium. For this, we consider the dynamical system \eqref{eq: dynamical system} where we replace the trait $y$ by $-x$.
	Then, using the symmetry of the involved functions and the traits as well as $\alpha(x,x)=1$, we find that the non-trivial equilibrium population satisfies \[
	n_1^a=n_2^a=\frac{\lambda(x)-\mu(x)}{(1+\alpha(-x,x))(1-p(x))}\quad\text{and}\quad n_1^d=n_2^d=\frac{(\lambda(x)-\mu(x))^2 p(x)}{\sigma(x)(1+\alpha(-x,x))(1-p(x))^2}.
	\]
	In our case, we obtain \[
	n_1^a=n_2^a=\frac{e^{-x^2/(2\sigma_b^2)+x^2/(2\sigma_p^2)}}{1+e^{-(2x)^2/(2\sigma_{\alpha}^2)}}\quad\text{and}\quad n_1^d=n_2^d=\frac{(e^{-x^2/(2\sigma_b^2)})^2(1-e^{-x^2/(2\sigma_p^2)})}{\sigma(1+e^{-(2x)^2/(2\sigma_{\alpha}^2)})(e^{-x^2/(2\sigma_p^2)})^2}.
	\]
	
	Now, an invading mutant trait $y$ is successful against this coexistence if and only if its invasion fitness is positive. We know from \cite[Section 2, Appendix A]{BPT21} that the invasion fitness in our setting can be calculated by the formula
	\begin{align*}
		h(y,x)=\frac{\rho(y,x)-\sigma+\sqrt{(\rho(y,x)+\sigma)^2+4\sigma p(y)(\alpha(y,x)+\alpha(y,-x))n_1^a}}{2}
	\end{align*}
	where $\rho(y,x)=\lambda(y)-(\alpha(y,x)+\alpha(y,-x))n_1^a$. This follows from the death rate in the approximating branching process being $d=\mu(y)+\alpha(y,x)n_1^a+\alpha(y,-x)n_2^a$. Now, we consider the invading trait to be close to $x>0$, say $y=x+\varepsilon$ with $\varepsilon>0$. Then the first positive zero of $h(x+\varepsilon,x)=\tilde{h}_\varepsilon(x)$ only depends on $\varepsilon$. Denote this zero by $z(\varepsilon)$. As the radius of mutation tends to zero, the fitness of a mutant trait can be approximated by $\tilde{h}_\varepsilon(x)$ with $\varepsilon\to 0$. In particular, the point at which there are no further successful invasions is given by $z=\lim_{\varepsilon\to 0}z(\varepsilon)$. Unfortunately, we cannot hope for a simple explicit form of $z$, as the zeros of $h$ or $\tilde{h}_\varepsilon$ respectively are not easy to compute. However, we can numerically approximate these endpoints of the branching. Similar computations have been conducted by \cite{SMJV13} in a related model without dormancy. Due to a simpler form of the fitness function, they were able to find an explicit formula.
	
	One could also expect the traits to optimize over time, such that the active population is maximised. In other words, the population will evolve towards the trait $x$ and $-x$ respectively for which the maximum of $n_1^a$ is achieved. In fact, this is the first point where \[
	\frac{d}{d x}n_1^a=0.
	\]
	However, while this idea is indeed correct in a population consisting of a single trait, as it is known that the invasion fitness of populations with competition induced dormancy is positive if and only if the active equilibrium population size is larger (cf. \cite[Eq 2.5]{BT20}), this is not the case in a polymorphic population. For the example presented in Figure \ref{Fig: Expdormsims-r=2-dist} we calculate the position of the traits before any secondary branchings to be $x\approx 0.897$ while the approach via maximising the active population size would yield $x\approx 1.10$. This shows that the coexistence of traits prevents the population to reach its maximal resource efficiency. Similarly, we calculate the position of the branches for the same parameters without dormancy to be $x\approx 0.589$. Hence, if one interprets the trait space as genetic distance, then dormancy may lead to an increase in this distance. Hence, if speciation only occurs when a certain threshold is surpassed, then dormancy can contribute to speciation in this sense.

\begin{figure}[htbp]
\begin{subfigure}[t]{0.45\textwidth}
	\hspace{-0.75cm}\begin{tikzpicture}
		\node (img1)  {\includegraphics[width=1.2\textwidth]{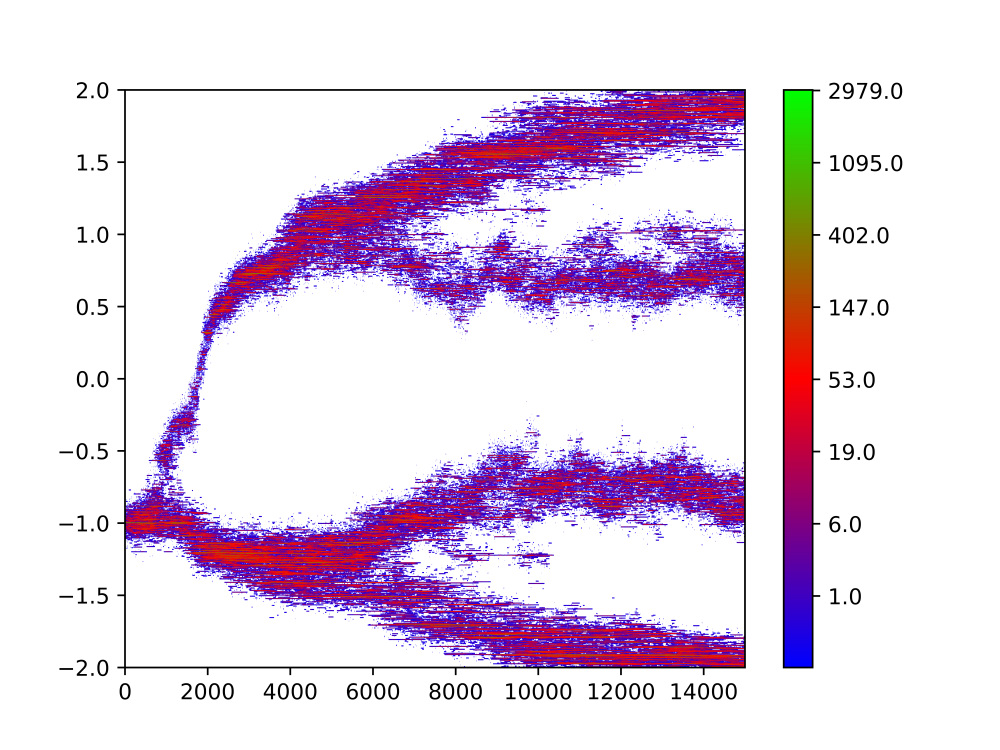}};
		\node[below=of img1, node distance=0cm, yshift=1.5cm, xshift=-0.5cm] {\footnotesize Time};
		\node[left=of img1, node distance=0cm, rotate=90, anchor=center,yshift=-1.5cm] {\footnotesize Trait};
	\end{tikzpicture}
	\subcaption{A simulation of the model with $r=1.5$, $\sigma_\alpha^2=1$, $\sigma_b^2=1.8$, $\sigma_p^2=1.7$, $\sigma=0.1$, $m=0.01$, $\rho^2=0.05$, $K=1000$.}
\end{subfigure}\hspace{0.7cm}
\begin{subfigure}[t]{0.45\textwidth}
	\hspace{-0.5cm}\begin{tikzpicture}
		\node (img1)  {\includegraphics[width=1.2\textwidth]{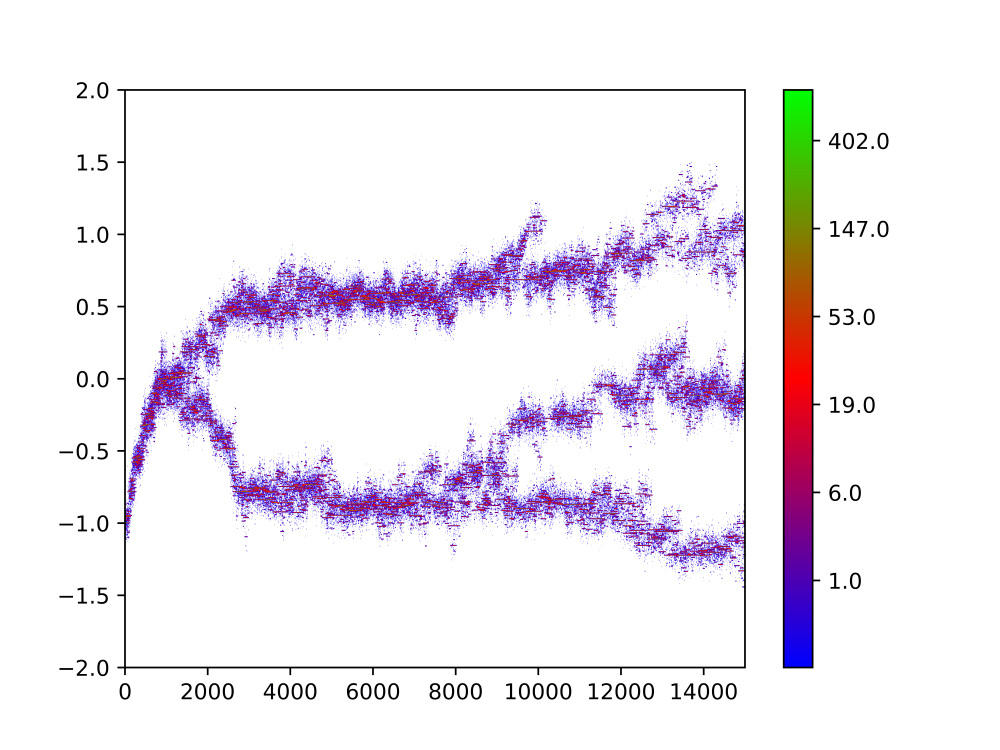}};
		\node[below=of img1, node distance=0cm, yshift=1.5cm, xshift=-0.5cm] {\footnotesize Time};
		\node[left=of img1, node distance=0cm, rotate=90, anchor=center,yshift=-1.5cm] {\footnotesize Trait};
	\end{tikzpicture}
	\subcaption{A simulation of the model without the dormancy mechanism and $\sigma_\alpha^2=2$, $\sigma_b^2=1.5$, $m=0.01$, $\rho^2=0.05$, $K=1000$.}
\end{subfigure}
	\caption{Simulations showing the population size of each trait over time. Initially, the population is composed of $250$ active individuals with trait $-1$. The branches with dormancy appear thicker and additional branches are produced.}
	\label{Fig: Expdormsims-r=15}
\end{figure}
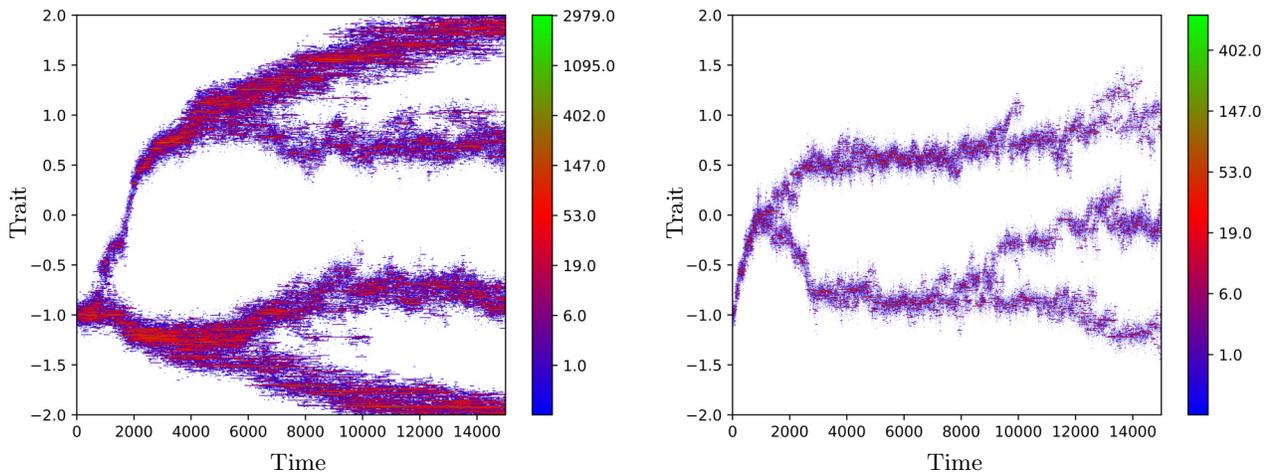

{\bf Dormancy increases the carrying capacity of sub-populations.}	
 We also compare the total population size with the example given in Figure \ref{Fig: Expdormsims-r=15}: when dormancy is involved, the environment can support significantly more individuals as is seen from the colour scale indicating the current size of the trait. This is due to the fact that dormant individuals do not use any resources. Another qualitative feature of these simulations is the thickness of each branch. We may think of the branches as different subspecies, the thickness of a branch indicates the genetic variation within a given subspecies or the width of the niche the given subspecies occupies. Here, the branches with dormancy occupy a larger range of the trait space, while the ones without dormancy appear narrow and light. This observation can also be made in Figure \ref{Fig: Expdormsims-r=2}. Note that this stems from a slow rate of resuscitation (on average once per ten units of time). Hence the population experiences less exposure to stochastic fluctuations in the trait. Mutations which are not necessarily advantageous remain in the population for longer due to the escape into dormancy. When the resuscitation rate is increased, this effect is reduced as can be seen in Figure \ref{Fig: Expdormsims-r=2-dist}. We can only confirm our observations for finite populations, since we cannot calculate the coexistence equilibria for an arbitrary number of traits and hence cannot give an exact simulation of the PES. However, it may be possible that even in the limit $K\to\infty$ we see wider ranges for coexistence.

{\bf Dormancy can create alternative paths to dimorphic populations.}		
	In the stochastic system, it may happen by chance that some individuals stay dormant for long enough until the general population sits at a large distance. When these individuals wake up, they may experience little competition and survive to create a new branch. We call this effect ``tunnelling'' and an example of this effect is displayed in Figure \ref{Fig: Tunnelsim}.
	
	\begin{figure}
		\begin{tikzpicture}
			\node (img1)  {\includegraphics[width=0.6\textwidth]{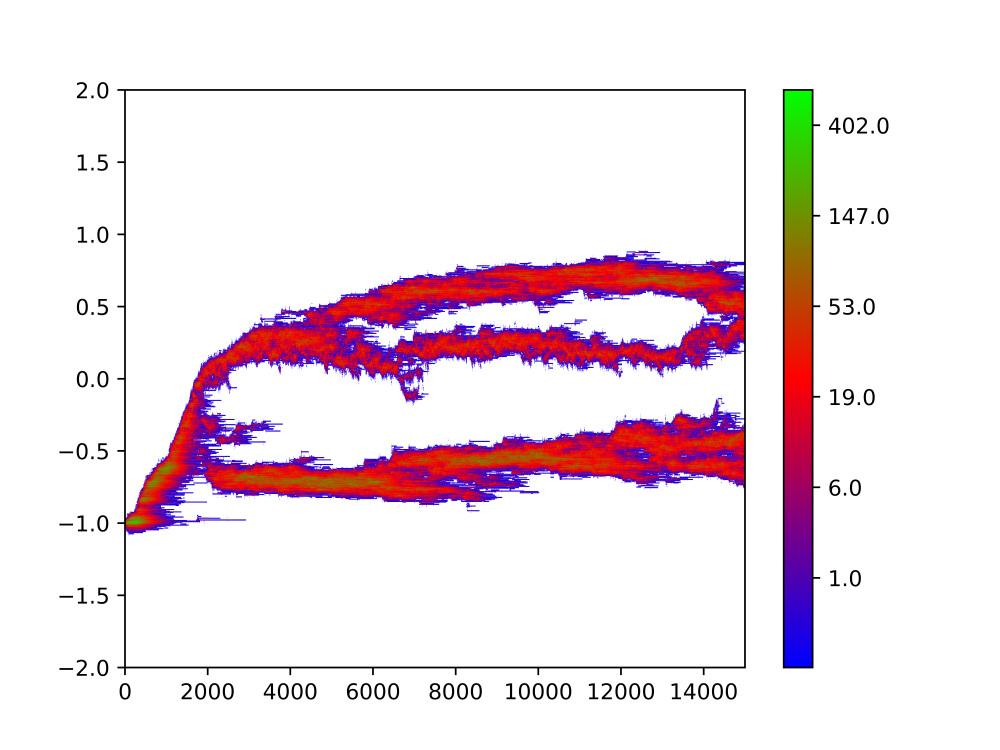}};
			\node[below=of img1, node distance=0cm, yshift=1.5cm, xshift=-0.5cm] {\footnotesize Time};
			\node[left=of img1, node distance=0cm, rotate=90, anchor=center,yshift=-1.5cm] {\footnotesize Trait};
		\end{tikzpicture}
		\caption{Simulation showing the population size of each trait over time. Initially, the population is composed of $250$ active individuals with trait $-1$. The parameters are $r=2$, $\sigma_b^2=0.7$, $\sigma_{\alpha}^2=0.49$, $\sigma_p^2=5$, $m=0.1$, $\sigma=0.01$, $\rho^2=0.01$ and $K=1000$. }
		\label{Fig: Tunnelsim}
	\end{figure}
	
	 Note however that this is not possible when we consider the scaling of $K\to\infty$ in the setting of the polymorphic evolution sequence because mutations are rare. In fact, the rarity of mutations prevents the survival of any dormant individuals against the dominating trait unless they can coexist. Hence, this is a purely stochastic effect. For one, we require a dormant individual to remain dormant for a sufficient period of time, which only depends on $\sigma$. Secondly, we require this individual to have a positive fitness at the time of resuscitation, which depends on the birth rate $\lambda$ and the competition experienced. The latter implicitly depends on the mutation parameter $m$, since more mutations allow the population to evolve faster into the reproduction optimum. Suppose that the population has evolved towards its reproduction optimum $0$ and a resuscitating ``mutant'' trait $y$ enters the active population. For convenience, we will assume that the trait $0$ has been adopted by the population and hence the equilibrium population size of a monomorphic population is assumed. Then, the invasion fitness of trait $y$ against trait $0$ is positive if and only if the branching criterion is satisfied. Hence, this mechanism of branching does not allow a larger range of parameters to observe a splitting of the population. In fact, we expect the two branches of the population to end up in the same regions of the trait space as the branches from evolutionary branching around $0$ would. However, it significantly changes the history of the population and it may allow for temporary additional branches as can be seen in Figure \ref{Fig: Tunnelsim}. While the behaviour in the stochastic setting after resuscitation of a subpopulation by tunnelling is difficult to understand, we have not observed a similar structure in simulations without such a tunnelling effect. A similar observation regarding the effects of stochasticity on qualitative aspects of evolutionary dynamics has been made by \cite{WI13}. In their model, they found that evolutionary branching only occurs for sufficiently large population sizes. This contrasts our result, since we require sufficiently small populations, but also demonstrates the importance of finite population sizes when analysing the models.

	\section*{Acknowledgements}
	This work was partially funded by the Deutsche Forschungsgemeinschaft (DFG, German Research Foundation) under Germany’s Excellence Strategy MATH+: The Berlin Mathematics Research Center, EXC-2046/1 project ID EF 4-7

\end{document}